\documentclass[12pt,english,american]{article}
\usepackage[T1]{fontenc}
\usepackage[latin9]{inputenc}
\usepackage{geometry}
\usepackage{color}
\usepackage{babel}
\usepackage{verbatim}
\usepackage{refstyle}
\usepackage{url}
\usepackage{amsthm}
\usepackage{amsmath}
\usepackage{amssymb}
\usepackage{graphicx}
\usepackage[unicode=true,pdfusetitle,
 bookmarks=true,bookmarksnumbered=false,bookmarksopen=false,
 breaklinks=true,pdfborder={0 0 0},backref=false,colorlinks=true]
 {hyperref}
\hypersetup{
 linkcolor=blue,citecolor=blue, urlcolor=blue}

\makeatletter

\AtBeginDocument{\providecommand\secref[1]{\ref{sec:#1}}}
\AtBeginDocument{\providecommand\defref[1]{\ref{def:#1}}}
\AtBeginDocument{\providecommand\thmref[1]{\ref{thm:#1}}}
\AtBeginDocument{\providecommand\corref[1]{\ref{cor:#1}}}
\AtBeginDocument{\providecommand\figref[1]{\ref{fig:#1}}}
\AtBeginDocument{\providecommand\lemref[1]{\ref{lem:#1}}}
\RS@ifundefined{subref}
  {\def\RSsubtxt{section~}\newref{sub}{name = \RSsubtxt}}
  {}
\RS@ifundefined{thmref}
  {\def\RSthmtxt{theorem~}\newref{thm}{name = \RSthmtxt}}
  {}
\RS@ifundefined{lemref}
  {\def\RSlemtxt{lemma~}\newref{lem}{name = \RSlemtxt}}
  {}

\usepackage{enumitem}		
\usepackage[color=green!40]{todonotes}
\theoremstyle{plain}
\newtheorem{thm}{\protect\theoremname}[section]
\theoremstyle{definition}
\newtheorem{defn}[thm]{\protect\definitionname}
\theoremstyle{remark}
\newtheorem{rem}[thm]{\protect\remarkname}
\theoremstyle{plain}
\newtheorem{lem}[thm]{\protect\lemmaname}
\ifx\proof\undefined
\newenvironment{proof}[1][\protect\proofname]{\par
\normalfont\topsep6\p@\@plus6\p@\relax
\trivlist
\itemindent\parindent
\item[\hskip\labelsep
\scshape
#1]\ignorespaces
}{%
\endtrivlist\@endpefalse
}
\providecommand{\proofname}{Proof}
\fi
\theoremstyle{plain}
\newtheorem{cor}[thm]{\protect\corollaryname}

\usepackage{txfonts}
\newref{con}{name = Conjecture~}
\newref{prop}{name = Proposition~}
\newref{def}{name = Definition~}
\newref{sec}{name = Section~}
\newref{thm}{name = Theorem~}
\newref{lem}{name = Lemma~}
\newref{cor}{name = Corollary~}
\newref{fig}{name = Figure~}

\@ifundefined{showcaptionsetup}{}{%
 \PassOptionsToPackage{caption=false}{subfig}}
\usepackage{subfig}
\makeatother

\addto\captionsamerican{\renewcommand{\corollaryname}{Corollary}}
\addto\captionsamerican{\renewcommand{\definitionname}{Definition}}
\addto\captionsamerican{\renewcommand{\lemmaname}{Lemma}}
\addto\captionsamerican{\renewcommand{\remarkname}{Remark}}
\addto\captionsamerican{\renewcommand{\theoremname}{Theorem}}
\addto\captionsenglish{\renewcommand{\corollaryname}{Corollary}}
\addto\captionsenglish{\renewcommand{\definitionname}{Definition}}
\addto\captionsenglish{\renewcommand{\lemmaname}{Lemma}}
\addto\captionsenglish{\renewcommand{\remarkname}{Remark}}
\addto\captionsenglish{\renewcommand{\theoremname}{Theorem}}
\providecommand{\corollaryname}{Corollary}
\providecommand{\definitionname}{Definition}
\providecommand{\lemmaname}{Lemma}
\providecommand{\remarkname}{Remark}
\providecommand{\theoremname}{Theorem}

\begin{document}

\title{Delay Equations of the Wheeler-Feynman Type}

\author{D.-A. Deckert, D. Dürr, N. Vona}

\maketitle
\begin{abstract}
We present an approximate model of Wheeler-Feynman electrodynamics
for which uniqueness of solutions is proved. It is simple enough to
be instructive but close enough to Wheeler-Feynman electrodynamics
such that we can discuss its natural type of initial data, constants
of motion and stable orbits with regard to Wheeler-Feynman electrodynamics.\\
\\
\textbf{Acknowledgements:} The authors want to thank Gernot Bauer
for enlightening discussions. D.-A.D. and N.V. gratefully acknowledge
financial support from the \emph{BayEFG} of the \emph{Freistaat Bayern}
and the \emph{Universität Bayern e.V.}. This project was also funded
by the post-doc program of the \emph{DAAD}.
\end{abstract}
\global\long\def\vect#1{\mathbf{#1}}

\section{Introduction}

Already in the early stages of classical electrodynamics it was observed
that the coupled equations of motion of Maxwell and Lorentz for point-charges
are ill-defined. Since then many attempts have been made to cure this
problem among which the most famous one is Dirac's mass renormalization
\cite{dirac_classical_1938}. However, none of the known cures has
yet led to a mathematically well-defined and physically sensible theory
of relativistic interaction of point-charges. Rather than being a
cure of the old theory, Wheeler-Feynman electrodynamics (WF) is a
new formulation of classical electrodynamics and is maybe the most
promising candidate for a divergence-free theory of electrodynamics
about point-charges that is capable of describing the phenomenon of
radiation damping \cite{wheeler_interaction_1945,wheeler_classical_1949}.
For an overview on the mathematical and physical difficulties of classical
electrodynamics and the role of WF we refer the interested reader
also to the introductory sections of \cite{bauer_maxwell-lorentz_2001}
and \cite{bauer_wheeler_2010}. In contrast to text-book electrodynamics,
which introduces fields as well as charges, WF is a theory only about
charges -- fields only occur as mathematical entities and not as dynamical
degrees of freedom. Let $N$ be the number of charges, $\vect q_{i,t}\in\mathbb{R}^{3}$
the position of the $i$-th charge at time $t$, $m_{i}>0$ its mass,
and 
\[
\vect p_{i,t}=\frac{m_{i}\vect v_{i,t}}{\sqrt{1-\vect v_{i,t}^{2}}},\qquad\vect v_{i,t}:=\frac{d\vect q_{i,t}}{dt}
\]
its relativistic momentum and velocity. The fundamental equations
of WF take the form

\begin{equation}
\frac{d}{dt}\begin{pmatrix}\vect q_{i,t}\\
\vect p_{i,t}
\end{pmatrix}=\begin{pmatrix}\vect v(\vect p_{i,t}):=\frac{\vect p_{i,t}}{\sqrt{m_{i}^{2}+\vect p_{i,t}^{2}}}\\
e_{i}\sum_{j\neq i}\vect F_{t}[\vect q_{j}](\vect q_{i,t},\vect p_{i,t})
\end{pmatrix},\qquad1\leq i\leq N.\label{eq:fundamental}
\end{equation}
We have chosen units such that the speed of light equals one. The
force term $\vect F_{t}[\vect q_{j}]$ is a functional of the trajectory
$t\mapsto\vect q_{j}$ and can be expressed by 
\[
\vect F_{t}[\vect q_{j}](\vect x,\vect p):=e_{j}\sum_{\pm}\left(\vect E_{t}^{\pm}[\vect q_{j}](\vect x)+\vect v(\vect p)\wedge\vect B_{t}^{\pm}[\vect q_{j}](\vect x)\right)
\]
where $(\vect E_{t}^{+}[\vect q_{j}],\vect B_{t}^{+}[\vect q_{j}])$
and $(\vect E_{t}^{-}[\vect q_{j}],\vect B_{t}^{-}[\vect q_{j}])$
denote the advanced and retarded Liénard-Wiechert fields -- in our
notation $\vect E$ represents the electric and $\vect B$ the magnetic
component and $\wedge$ is the outer product. The Liénard-Wiechert
fields are special solutions to the Maxwell equations corresponding
to a prescribed point-charge trajectory $t\mapsto\vect q_{j,t}$;
see e.g. \cite{bauer_wheeler_2010}. Their explicit form is
\begin{eqnarray}
\vect E_{t}^{\pm}[\vect q_{j}](\vect x) & := & \left[\frac{(\vect n_{j}\pm\vect v_{j})(1-\vect v_{j}^{2})}{\left\Vert \vect x-\vect q_{j}\right\Vert ^{2}(1\pm\vect n_{j}\cdot\vect v_{j})^{3}}+\frac{\vect n_{j}\wedge\left(\vect n_{j}\wedge\vect a_{j}\right)}{\left\Vert \vect x-\vect q_{j}\right\Vert (1\pm\vect n_{j}\cdot\vect v_{j})^{3}}\right]_{\pm},\label{eq:LW_E}\\
\vect B_{t}^{\pm}[\vect q_{j}](\vect x) & := & \mp\vect n_{j,\pm}\wedge\vect E_{t}^{\pm}[\vect{q_{j}}](\vect x).\nonumber 
\end{eqnarray}
Here, we have used the abbreviations
\begin{equation}
\vect q_{j,\pm}=\vect q_{j,t^{\pm}},\qquad\vect n_{j,\pm}=\frac{\vect x-\vect q_{j,t_{j}^{\pm}}}{\left\Vert \vect x-\vect q_{j,t_{j}^{\pm}}\right\Vert },\qquad\vect v_{j}=\frac{d\vect q_{j,t}}{dt}\Bigg|_{t=t_{j}^{\pm}},\qquad\vect a_{j}=\frac{d^{2}\vect q_{j,t}}{dt^{2}}\Bigg|_{t=t_{j}^{\pm}},\label{eq:LW_notation}
\end{equation}
and the delayed times $t_{j}^{+}$ and $t_{j}^{-}$ are implicitly
defined as solutions to 
\begin{equation}
t_{j}^{\pm}(t,\vect x):=t\pm\left\Vert \vect x-\vect q_{j,t_{j}^{\pm}(t,\vect x)}\right\Vert .\label{eq:wf_delay}
\end{equation}
Given the space-time point $(t,\vect x)$ the delayed times $t_{j}^{+}$
and $t_{j}^{-}$ are determined by the intersection points of the
trajectory $t\mapsto\vect q_{j,t}$ with the future and past light-cone
of $(t,\vect x)$, respectively. As long as $\left\Vert \dot{\vect q}_{j,t}\right\Vert <1$
both times $t_{j}^{+}$ and $t_{j}^{-}$ are well-defined. As it becomes
apparent from (\ref{eq:wf_delay}) the WF equations (\ref{eq:fundamental})
involve terms with time-like advanced as well as retarded arguments
which makes them mathematically very hard to handle: 
\begin{itemize}
\item The question of existence of solutions is completely open with the
exception of the following special cases: Schild found explicit solutions
formed by two attracting charges that revolve around each other on
stable orbits \cite{schild_electromagnetic_1963}. Driver proved existence
and uniqueness of solutions for two repelling charges constrained
to the straight line whenever initially the relative velocity is zero
and the spatial separation is sufficiently large \cite{driver_canfuture_1979}.
Furthermore, Bauer proved existence of solutions in the case of two
repelling charges constrained to the straight line \cite{bauer_ein_1997},
and a more recent result ensures the existence of solutions on finite
but arbitrarily large time intervals for $N$ arbitrary charges in
three spatial dimension \cite{bauer_wheeler_2010}.
\item It is not even clear what can be considered to be sensible initial
data to uniquely identify WF solutions. While Driver's uniqueness
result suggests the specification of initial positions and momenta
of the charges, Bauer's work points to the use of asymptotic positions
and velocities to distinguish scattering solutions, and one sentence
below Figure 3 in \cite{wheeler_classical_1949} hints to a certain
configuration of whole trajectory strips of the charges as initial
conditions.
\end{itemize}
Furthermore, one important issue of WF is the justification of the
so called absorber condition and the derivation of the irreversible
behavior which we experience in radiation phenomena \cite{wheeler_interaction_1945}.
The link between WF and experience must be based on a notion of typical
trajectories, i.e. on a measure of typicality for WF dynamics. Such
a measure is unknown and at the moment out of reach. A generalization
to many particles of the approximate model we consider next provides
an excellent case study for introducing a notion of typicality for
this kind of non-markovian dynamics. Note that uniqueness as well
as conservation of energy are very likely to be important for defining
a measure of typicality in the sense of Boltzmann. This is work in
progress.\\

Our intention behind this work is to provide a pedagogical introduction
to the mathematical structures arising from delay differential equations
of the WF type by discussing initial data and uniqueness, constants
of motion and stable orbits in a non-trivial approximate model of
WF, which was introduced in \cite{deckert_electrodynamic_2010}, and
for which many mathematical questions can be answered satisfactorily.
Compared to (\ref{eq:fundamental}) we make the following simplifications:
\begin{itemize}
\item We consider only two charges, i.e. $N=2.$
\item As fundamental equations of this approximate model we take the form
(\ref{eq:fundamental}) where we replace the Liénard-Wiechert fields
$(\vect E_{t}^{\pm}[\vect q_{j}](\vect x),\vect B_{t}^{\pm}[\vect q_{j}](\vect x))$
with
\begin{equation}
\vect E_{t}^{\parallel,\pm}[\vect q_{j}](\vect x):=e_{j}\frac{\vect x-\vect q_{j,t_{j}^{\pm}(t,\vect x)}}{\left\Vert \vect x-\vect q_{j,t_{j}^{\pm}(t,\vect x)}\right\Vert ^{3}},\qquad\vect B_{t}^{\pm}[\vect q_{j}]=0,\label{eq:new fields}
\end{equation}
i.e. the Coulomb fields at the respective delayed times in the future
and the past.
\end{itemize}
The resulting equations are
\begin{equation}
\frac{d}{dt}\begin{pmatrix}\vect q_{i,t}\\
\vect p_{i,t}
\end{pmatrix}=\begin{pmatrix}\vect v(\vect p_{i,t})\\
e_{i}e_{j}\left[\vect F\left(\vect q_{i,t}-\vect q_{j,t_{j}^{+}(t,\vect q_{i,t})}\right)+\vect F\left(\vect q_{i,t}-\vect q_{j,t_{j}^{-}(t,\vect q_{i,t})}\right)\right]
\end{pmatrix},\qquad i,j\in\{1,2\},\quad j\neq i,\label{eq:toy model}
\end{equation}
where the force field is given by
\begin{equation}
\vect F:\mathbb{R}^{3}\setminus\{0\}\to\mathbb{R}^{3}\setminus\{0\},\qquad\vect x\mapsto\vect F(\vect x):=\frac{\vect x}{\left\Vert \vect x\right\Vert ^{3}}.\label{eq:force field}
\end{equation}

We emphasize that the delay function (\ref{eq:wf_delay}) is the same
as in WF and that the simplified fields (\ref{eq:new fields}) are
the longitudinal modes of the Liénard-Wiechert fields $(\vect E_{t}^{\pm}[\vect q_{j}](\vect x),\vect B_{t}^{\pm}[\vect q_{j}](\vect x))$,
i.e.
\[
\nabla\cdot\left(\vect E_{t}^{\pm}[\vect q_{j}](\vect x)-\vect E_{t}^{\parallel,\pm}[\vect q_{j}](\vect x)\right)=0,\qquad\nabla\wedge\vect E_{t}^{\parallel,\pm}[\vect q_{j}](\vect x)=0.
\]
Furthermore, for small velocities and accelerations of the $j$-th
charge one finds $\vect E_{t}^{\pm}[\vect q_{j}](\vect x)\approx\vect E_{t}^{\parallel,\pm}[\vect q_{j}](\vect x)$.
In this sense one can regard this simplified model as an physically
interesting approximation to WF. Note also that, in contrast to $\vect E_{t}^{\pm}[\vect q_{j}]$,
the simplified field $\vect E_{t}^{\parallel,\pm}[\vect q_{j}]$ does
not depend on the acceleration of the $j$-th charge; compare (\ref{eq:LW_E}).
This fact makes it easier to control smoothness of solutions, however,
is not the key difference that allows us to prove uniqueness of solutions
for the approximate model. \\
\\
This paper is organized as follows: 
\begin{enumerate}
\item In \secref{Construction-of-Solutions} we discuss natural initial
data for the approximate model (\defref{initial_data}), and show
how to construct solutions uniquely depending on given initial data
(\thmref{uniquenss}). A byproduct is the observation that in general
the specification of initial positions and momenta of the two charges
is not sufficient to ensure uniqueness (\corref{newtonian}).
\item In \secref{Constants-of-Motion} we discuss the structure of constants
of motion by reference to the energy conservation in the approximate
model. The energy functional is defined in \defref{energy} and the
energy conservation is shown in \thmref{energy}.
\item In \secref{Stable-Orbits} we identify stable orbits (\defref{schild})
which solve (\ref{eq:toy model}) as proven in \thmref{schild}.
\item We conclude in \secref{Conclusion} by putting these results in perspective
to WF.
\end{enumerate}

\paragraph{Notation.}
\begin{itemize}
\item If not otherwise specified we use the convention that $i,j\in\{1,2\}$
and $j\neq i$.
\item Any derivative at a boundary of an interval is understood as the left-hand
or right-hand side derivative.
\item Vectors in $\vect x,\vect y\in\mathbb{R}^{3}$ are denoted by bold
letters, the inner product by $\vect x\cdot\vect y$, the outer product
by $\vect x\wedge\vect y$, and the euclidean norm by $\left\Vert \vect x\right\Vert $.
\item $\nabla$, $\nabla\cdot$, and $\nabla\wedge$ denote the gradient,
divergence and curl w.r.t. $\vect x$, respectively.
\item Overset dots denote derivatives with respect to time, i.e. $\dot{\vect q}_{i,t}=\frac{d\vect q_{i,t}}{dt}$
and $\ddot{\vect q}_{i,t}=\frac{d^{2}\vect q_{i,t}}{dt^{2}}$.
\item The future and past light-cone of the space-time point $(t,\vect x)$
is defined as the set 
\[
L_{\pm}(t,\vect x):=\{(s,\vect y)\,|\,(t-s)^{2}-(\vect x-\vect y)^{2}=0,t>\pm s\}
\]
 for $+$ and $-$, respectively.
\end{itemize}

\section{Construction of Solutions\label{sec:Construction-of-Solutions}}

For this section we regard the indices $(i,j)$ fixed to either $(1,2)$
or $(2,1)$. We consider the following initial data (see \figref{Initial-data.}):

\begin{figure}
\hfill{}\subfloat[\label{fig:Initial-data.}The thick black stripes represent the initial
data.]{\includegraphics{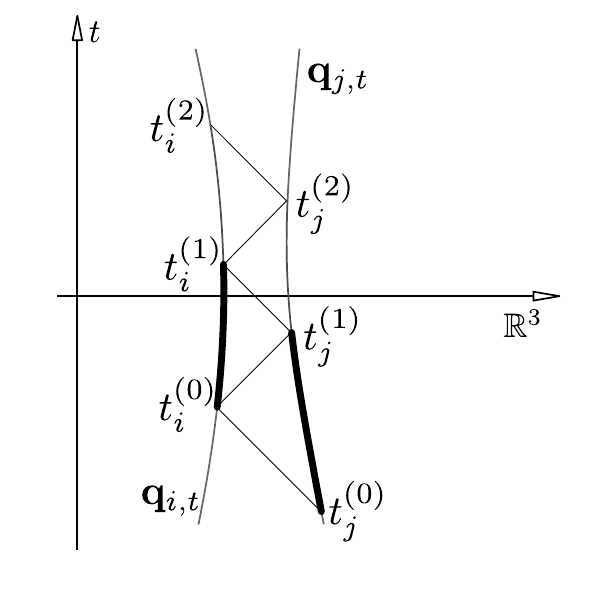}}\hfill{}\subfloat[\label{fig:Nature-of-the}Nature of the interaction.]{\includegraphics{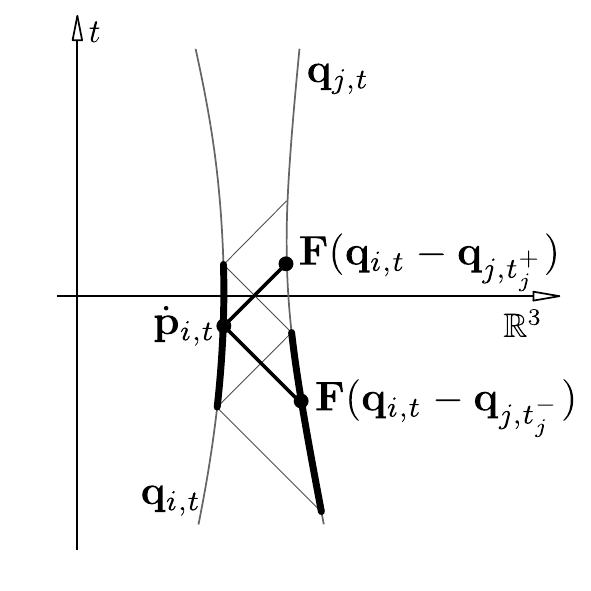}}\hfill{}

\caption{}
\hfill{}
\end{figure}

\begin{defn}
\label{def:initial_data}We call a pair $\left(\vect q_{i}^{(0)},\vect p_{i}^{(0)}\right)_{i=1,2}$
of smooth position and momentum trajectory stripes
\begin{eqnarray*}
\vect q_{i}^{(0)}: & \mathbb{R}\supset[t_{i}^{(0)},t_{i}^{(1)}] & \to\mathbb{R}^{3},\qquad t\mapsto\vect q_{i,t}^{(0)},\\
\vect p_{i}^{(0)}: & \mathbb{R}\supset[t_{i}^{(0)},t_{i}^{(1)}] & \to\mathbb{R}^{3},\qquad t\mapsto\vect p_{i,t}^{(0)},\qquad i\in\{1,2\}
\end{eqnarray*}
that fulfill the following conditions \emph{initial data} for equation
(\ref{eq:toy model}):
\begin{enumerate}[label=(\roman*)]
\item For times $t\in[t_{i}^{(0)},t_{i}^{(1)}]\cap[t_{j}^{(0)},t_{j}^{(1)}]$
it holds
\[
\vect q_{i,t}^{(0)}\neq\vect q_{j,t}^{(0)}.
\]

\item For all $t\in[t_{i}^{(0)},t_{i}^{(1)}]$ it holds
\begin{equation}
\frac{d}{dt}\vect q_{i,t}^{(0)}=\vect v(\vect p_{i,t}^{(0)})\equiv\frac{\vect p_{i,t}^{(0)}}{\sqrt{m_{i}^{2}+\left\Vert \vect p_{i,t}^{(0)}\right\Vert ^{2}}},\qquad i\in\{1,2\}.\label{eq:velocity}
\end{equation}

\item The times $t_{i}^{(0)},t_{i}^{(1)},t_{j}^{(0)},t_{j}^{(1)}$ relate
to each other according to
\begin{equation}
t_{i}^{(0)}=t_{i}^{+}\left(t_{j}^{(0)},\vect q_{j,t_{j}^{(0)}}^{(0)}\right),\qquad t_{j}^{(1)}=t_{j}^{+}\left(t_{i}^{(0)},\vect q_{i,t_{i}^{(0)}}^{(0)}\right),\qquad t_{i}^{(1)}=t_{i}^{+}\left(t_{j}^{(1)},\vect q_{j,t_{j}^{(1)}}^{(0)}\right).\label{eq:times}
\end{equation}

\item At time $t=t_{i}^{(0)}$and for all integers $n\geq0$ the trajectories
obey \foreignlanguage{english}{
\begin{equation}
\frac{d^{n}}{dt^{n}}\begin{pmatrix}\dot{\vect q}_{i,t}\\
\dot{\vect p}_{i,t}
\end{pmatrix}=\frac{d^{n}}{dt^{n}}\begin{pmatrix}\vect v(\vect p_{i,t})\\
e_{i}e_{j}\left[\vect F\left(\vect q_{i,t}-\vect q_{j,t_{j}^{+}(t,\vect q_{i,t})}\right)+\vect F\left(\vect q_{i,t}-\vect q_{j,t_{j}^{-}(t,\vect q_{i,t})}\right)\right]
\end{pmatrix}.\label{eq:fulfilled i}
\end{equation}
}
\item At time $t=t_{j}^{(1)}$and for all integers $n\geq0$ the trajectories
obey \foreignlanguage{english}{
\begin{equation}
\frac{d^{n}}{dt^{n}}\begin{pmatrix}\dot{\vect q}_{j,t}\\
\dot{\vect p}_{j,t}
\end{pmatrix}=\frac{d^{n}}{dt^{n}}\begin{pmatrix}\vect v(\vect p_{j,t})\\
e_{j}e_{i}\left[\vect F\left(\vect q_{j,t}-\vect q_{i,t_{i}^{+}(t,\vect q_{j,t})}\right)+\vect F\left(\vect q_{j,t}-\vect q_{i,t_{i}^{-}(t,\vect q_{j,t})}\right)\right]
\end{pmatrix}.\label{eq:fulfilled j}
\end{equation}
}
\end{enumerate}
\end{defn}
\begin{rem}
Note that (\ref{eq:velocity}) requires $\left\Vert \dot{\vect q}_{i,t}\right\Vert <1$
for $i\in\{1,2\}$ and 
\[
\vect p_{i,t}=\frac{m_{i}\dot{\vect q}_{i,t}}{\sqrt{1-\dot{\vect q}_{i,t}^{2}}}.
\]
Furthermore, such initial data can be constructed as follows: 
\begin{enumerate}
\item Choose times $t_{j}^{(0)}<t_{j}^{(1)}$, an arbitrary smooth trajectory
strip $\vect q_{j}^{(0)}:[t_{j}^{(0)},t_{j}^{(1)}]\to\mathbb{R}^{3}$
with $\left\Vert \dot{\vect q}_{j,t}^{(0)}\right\Vert <1$, one space-point
$(t_{i}^{(0)},\vect q_{i,t_{i}^{(0)}}^{(0)})$ on the intersection
of the forward light-cone of $(t_{j}^{(0)},\vect q_{j,t_{j}^{(0)}}^{(0)})$
with the backward light-cone of $(t_{j}^{(1)},\vect q_{j,t_{j}^{(1)}}^{(0)})$,
and one space-point $(t_{i}^{(1)},\vect q_{i,t_{i}^{(1)}}^{(0)})$
somewhere on the forward light-cone of $(t_{j}^{(1)},\vect q_{j,t_{j}^{(1)}}^{(0)})$
and inside the forward light-cone of $(t_{i}^{(0)},\vect q_{i,t_{i}^{(0)}}^{(0)})$.
\item Define $\vect p_{j,t}^{(0)}:=\frac{m_{j}\dot{\vect q}_{j,t}^{(0)}}{\sqrt{1-\dot{\vect q}_{j,t}^{2}}}$
for $t\in[t_{j}^{(0)},t_{j}^{(1)}]$ and, using (\ref{eq:toy model}),
compute all derivatives of $\dot{\vect p}_{i,t}$ at the times $t=t_{i}^{(0)}$
and $t=t_{i}^{(1)}$.
\item Choose a smooth trajectory strip $\vect q_{i}^{(0)}:[t_{i}^{(0)},t_{i}^{(1)}]\to\mathbb{R}^{3}$
through the space-time points $(t_{i}^{(0)},\vect q_{i,t_{i}^{(0)}}^{(0)})$
and $(t_{i}^{(1)},\vect q_{i,t_{i}^{(1)}}^{(0)})$ such that $\left\Vert \dot{\vect q}_{j,t}^{(0)}\right\Vert <1$
and that $t\mapsto\vect p_{i,t}^{(0)}:=\frac{m_{i}\dot{\vect q}_{i,t}^{(0)}}{\sqrt{1-\left\Vert \dot{\vect q}_{i,t}\right\Vert ^{2}}}$
smoothly connects to the derivatives computed in step 2.
\end{enumerate}
\end{rem}
Provided such initial data we shall prove our first result:
\begin{thm}
\label{thm:uniquenss}Given the initial data $\left(\vect q_{i}^{(0)},\vect p_{i}^{(0)}\right)_{i=1,2}$
there exist two smooth maps
\begin{equation}
\mathbb{R}\supseteq D_{i}\to\mathbb{R}^{3}\times\mathbb{R}^{3},\qquad t\mapsto(\vect q_{i,t},\vect p_{i,t}),\qquad i\in\{1,2\},\label{eq:solution}
\end{equation}
such that:
\begin{enumerate}[label=(\roman*)]
\item $(\vect q_{i,t},\vect p_{i,t})=(\vect q_{i,t}^{(0)},\vect p_{i,t}^{(0)})$
for all \foreignlanguage{english}{\textup{$t\in[t_{i}^{(0)},t_{i}^{(1)}]$}}
and $i\in\{1,2\}$.
\item $t\mapsto(\vect q_{i,t},\vect p_{i,t})$ solves (\ref{eq:toy model})
on $D_{i}:=\left(t_{i}^{+}(T_{j}^{\min},\vect q_{j,T_{j}^{\min}}),t_{i}^{-}(T_{j}^{\max},\vect q_{j,T_{j}^{\max}})\right)$
for $i\in\{1,2\}$, $j\neq i$. 
\item For $i=1,2$ let $\widetilde{D}_{i}\subseteq\mathbb{R}$ be an interval
such that $[t_{i}^{(0)},t_{i}^{(1)}]\subseteq\widetilde{D}_{i}$ and
let $\widetilde{D}_{i}\to\mathbb{R}^{3}\times\mathbb{R}^{3}$, $t\mapsto(\widetilde{\vect q}_{i,t},\widetilde{\vect p}_{i,t})$
be a smooth map such that $t\mapsto(\widetilde{\vect q}_{i,t},\widetilde{\vect p}_{i,t})$
solves (\ref{eq:toy model}) for $t\in\widetilde{D}_{i}$. Then: 
\[
(\widetilde{\vect q}_{i,t},\widetilde{\vect p}_{i,t})=(\vect q_{i,t},\vect p_{i,t})\,\forall t\in D_{i}\cap\widetilde{D}_{i},i\in\{1,2\}\]
\[\Updownarrow\]
\[(\widetilde{\vect q}_{i,t},\widetilde{\vect p}_{i,t})=(\vect q_{i,t},\vect p_{i,t})\,\forall t\in[t_{i}^{(0)},t_{i}^{(1)}],i\in\{1,2\}.\]

\end{enumerate}
Given appropriate constants $d>0$ and  $0\leq v<1$, the times $-\infty\leq T_{i}^{\min}\leq t_{i}^{(0)}<t_{i}^{(1)}\leq T_{i}^{\max}\leq\infty$
are defined such that $[T_{i}^{\min},T_{i}^{\max}]$ is the largest
interval containing $[t_{i}^{(0)},t_{i}^{(1)}]$ with the property:
\begin{equation}
\left\Vert \dot{\vect q}_{i,t}\right\Vert \leq v,\qquad\left\Vert \vect q_{i,t}-\vect q_{j,t_{j}^{\pm}(t,\vect q_{i,t})}\right\Vert \geq d,\qquad\forall\, t\in[T_{i}^{\min},T_{i}^{\max}]\label{eq:no singularity}
\end{equation}
Their value is determined during the construction of (\ref{eq:solution})
in the proof.\end{thm}
\begin{rem}
Our focus lies on the uniqueness assertion (iii) of \thmref{uniquenss}.
We do not attempt to give a priori bounds on $T_{i}^{\min},T_{i}^{\max}$,
$i\in\{1,2\}$, whose values are determined during the dynamics by
condition (\ref{eq:no singularity}). This condition is needed to
prevent two types of singularities that can occur: First, the approach
of the speed of light, and second, collision or infinitesimal approach
of charges. These singularities can also be present in WF which can
be seen directly from the form of the fields (\ref{eq:LW_E}). While
the second one is familiar since it is of the same type as seen in
the $N$-body problem of Newtonian gravitation \cite{siegel_lectures_1971},
the first one is very specific to WF-type delay problems. Such singularities
are due to the nature of the delay times $t_{j}^{+}$ and $t_{j}^{-}$
defined in (\ref{eq:wf_delay}) which tend to plus or minus infinity
if the $j$-th charge approaches the speed of light in the future
or the past, respectively; the origin of this singularity can be seen
best in (\ref{eq:d_delay}). Because of angular momentum conservation
it is however expected that for $N=2$ charges one always finds $T_{\min}=-\infty$
and $T_{\max}=\infty$, which is at least true for the solutions given
in \secref{Stable-Orbits}. A treatment of the $N$-body problem will
require a notion of typicality of solutions.
\end{rem}
The key ingredient of the proof, which can be checked by direct computation,
is the following:
\begin{lem}
\label{lem:key_ingredient}The following statements are true:
\begin{enumerate}[label=(\roman*)]
\item The map $\vect F$ defined in (\ref{eq:force field}) is bijective
and its inverse is given by
\[
\vect I:\mathbb{R}^{3}\setminus\{0\}\to\mathbb{R}^{3}\setminus\{0\},\qquad\vect y\mapsto\vect I(\vect y):=\frac{\vect y}{\left\Vert \vect y\right\Vert ^{3/2}}.
\]

\item Let $\left(\vect q_{i}^{(0)},\vect p_{i}^{(0)}\right)_{i=1,2}$ be
given initial data. For any integer $n\geq0$ equations (\ref{eq:fulfilled i})
and (\ref{eq:fulfilled j}) are equivalent to
\[
\frac{d^{n}}{dt^{n}}\left(\vect q_{i,t}-\vect q_{j,t_{j}^{\pm}(t,\vect q_{i,t})}\right)=\frac{d^{n}}{dt^{n}}\vect I\left(\frac{1}{e_{i}e_{j}}\dot{\vect p}_{i,t}-\vect F\left(\vect q_{i,t}-\vect q_{j,t_{j}^{\mp}(t,\vect q_{i,t})}\right)\right)\qquad\text{for }t=t_{i}^{(0)},
\]
and
\[
\frac{d^{n}}{dt^{n}}\left(\vect q_{j,t}-\vect q_{i,t_{i}^{\pm}(t,\vect q_{j,t})}\right)=\frac{d^{n}}{dt^{n}}\vect I\left(\frac{1}{e_{j}e_{i}}\dot{\vect p}_{j,t}-\vect F\left(\vect q_{j,t}-\vect q_{i,t_{i}^{\mp}(t,\vect q_{j,t})}\right)\right)\qquad\text{for }t=t_{j}^{(0)},
\]
respectively.
\end{enumerate}
\end{lem}
\lemref{key_ingredient} ensures that for example in situations as
depicted in \figref{Nature-of-the} we can compute from $\vect F(\vect q_{i,t}-\vect q_{j,t_{j}^{+}})$,
which is determined by the initial data and (\ref{eq:toy model}),
the space-time point $(t_{j}^{+},\vect q_{j,t_{j}^{+}})$. This is
the key ingredient in our construction:
\begin{proof}[Proof of \thmref{uniquenss}]
As a first step, we construct a smooth extension of $\vect q_{j}^{(0)}$
beyond time $t_{j}^{(1)}$. Let us introduce the short-hand notation
\[
t_{j}^{\pm}=t_{j}^{\pm}(t,\vect q_{i,t}).
\]
In general, any solution $t\mapsto(\vect q_{i,t},\vect p_{i,t})_{i=1,2}$
to (\ref{eq:toy model}) has to fulfill
\begin{equation}
\dot{\vect p}_{i,t}=e_{i}e_{j}\left[\vect F\left(\vect q_{i,t}-\vect q_{j,t_{j}^{+}}\right)+\vect F\left(\vect q_{i,t}-\vect q_{j,t_{j}^{-}}\right)\right].\label{eq:pdot}
\end{equation}
With the help of \lemref{key_ingredient}(i) we can bring this equation
into the form
\begin{equation}
\vect q_{j,t_{j}^{+}}=\vect q_{i,t}-\vect I\left(\frac{1}{e_{i}e_{j}}\dot{\vect p}_{i,t}-\vect F\left(\vect q_{i,t}-\vect q_{j,t_{j}^{\mp}(t,\vect q_{i,t})}\right)\right)\label{eq:inverse equation}
\end{equation}
for times $t\in[T_{i}^{\min},T_{i}^{\max}]$ because then
\[
\left\Vert \vect I\left(\frac{1}{e_{i}e_{j}}\dot{\vect p}_{i,t}-\vect F\left(\vect q_{i,t}-\vect q_{j,t_{j}^{-}}\right)\right)\right\Vert =\left\Vert \vect q_{i,t}-\vect q_{j,t_{j}^{\pm}}\right\Vert \geq d
\]
is guaranteed, and hence, the right-hand side of (\ref{eq:inverse equation})
is well-defined. We now make use of (\ref{eq:inverse equation}) to
compute
\[
\vect q_{j}^{(1)}:t\mapsto\vect q_{j,t}^{(1)}
\]
according to
\begin{eqnarray}
t_{j}^{+} & = & t+\left\Vert \vect I\left(\frac{1}{e_{i}e_{j}}\dot{\vect p}_{i,t}^{(0)}-\vect F\left(\vect q_{i,t}^{(0)}-\vect q_{j,t_{j}^{-}(t,\vect q_{i,t})}^{(0)}\right)\right)\right\Vert ,\label{eq:t j plus}\\
\vect q_{j,t_{j}^{+}}^{(1)} & = & \vect q_{i,t}^{(0)}-\vect I\left(\frac{1}{e_{i}e_{j}}\dot{\vect p}_{i,t}^{(0)}-\vect F\left(\vect q_{i,t}^{(0)}-\vect q_{j,t_{j}^{-}(t,\vect q_{i,t})}^{(0)}\right)\right)\label{eq:q j plus}
\end{eqnarray}
for all $t\in(t_{i}^{(0)},t_{i}^{(1)}]\cap(T_{i}^{\min},T_{i}^{\max})$.
Note that due to (\ref{eq:times}) the right-hand side of (\ref{eq:t j plus})
and (\ref{eq:q j plus}) is well-defined. We define
\[
t_{j}^{(2)}=\min\left\{ t_{j}^{+}|_{t=t_{i}^{(1)}},T_{j}^{\max}\right\} .
\]
Since $\vect q_{i}^{(0)}$ is smooth on $[t_{i}^{(0)},t_{i}^{(1)}]$
also $t_{j}^{+}$ depends smoothly on $t\in[t_{i}^{(0)},t_{i}^{(1)}]$,
and in consequence, $\vect q_{j}^{(1)}$ is smooth on $(t_{j}^{(1)},t_{j}^{(2)}]$.
Furthermore, a direct computation gives 
\begin{equation}
t\mapsto\frac{dt_{j}^{+}}{dt}=\frac{1+\vect n_{j,+}\cdot\dot{\vect q}_{i,t}^{(0)}}{1+\vect n_{j,+}\cdot\dot{\vect q}_{j,t^{+}}^{(1)}},\qquad\vect n_{j,+}:=\frac{\vect{\vect q}_{i,t}^{(0)}-\vect q_{j,t^{+}}^{(1)}}{\left\Vert \vect{\vect q}_{i,t}^{(0)}-\vect q_{j,t^{+}}^{(1)}\right\Vert }.\label{eq:d_delay}
\end{equation}
so that, using the notation 
\[
\frac{d}{dt_{j}^{+}}=\frac{dt}{dt_{j}^{+}}\frac{d}{dt},
\]
we may then compute 
\begin{eqnarray}
\lim_{s\searrow t_{j}^{(1)}}\frac{d^{n}}{ds^{n}}\vect q_{j,s}^{(1)} & = & \lim_{t\searrow t_{i}^{(0)}}\frac{d^{n}}{dt_{j}^{+n}}\vect q_{j,t_{j}^{+}}^{(1)}=\lim_{t\searrow t_{i}^{(0)}}\frac{d^{n}}{dt_{j}^{+n}}\left[\vect q_{i,t}^{(0)}-\vect I\left(\frac{1}{e_{i}e_{j}}\dot{\vect p}_{i,t}^{(0)}-\vect F\left(\vect q_{i,t}^{(0)}-\vect q_{j,t_{j}^{-}(t,\vect q_{i,t})}^{(1)}\right)\right)\right]\label{eq:check smoothness}
\end{eqnarray}
for every integer $n\geq0$. \lemref{key_ingredient}(ii) ensures
that
\[
(\ref{eq:check smoothness})=\frac{d^{n}}{ds^{n}}\vect q_{j,s}^{(0)}\big|_{s=t_{j}^{(1)}},
\]
and hence,
\[
t\mapsto\begin{cases}
\vect q_{j,t}^{(0)} & \text{for }t\in[t_{j}^{(0)},t_{j}^{(1)}]\\
\vect q_{j,t}^{(1)} & \text{for }t\in(t_{j}^{(1)},t_{j}^{(2)}]
\end{cases}
\]
is a smooth map on $[t_{j}^{(0)},t_{j}^{(2)})$. Furthermore, 
\begin{equation}
\vect p_{j,t}^{(1)}:=m\frac{\dot{\vect q}_{j,t}}{\sqrt{1-\dot{\vect q}_{j,t}^{2}}},\qquad\forall\, t\in(t_{j}^{(1)},t_{j}^{(2)}].\label{eq:p j 1}
\end{equation}
is well-defined by (\ref{eq:no singularity}).

In the second step, we use the analogous construction to extend $\vect q_{i}^{(0)}$
smoothly beyond time $t_{i}^{(1)}$: We define
\[
\vect q_{i}^{(1)}:t\mapsto\vect q_{i,t}^{(1)}
\]
by
\begin{eqnarray}
t_{i}^{+} & = & t+\left\Vert \vect q_{j,t}^{(1)}-\vect I\left(\frac{1}{e_{j}e_{i}}\dot{\vect p}_{j,t}^{(1)}-\vect F\left(\vect q_{j,t}^{(1)}-\vect q_{i,t_{i}^{-}(t,\vect q_{j,t})}^{(0)}\right)\right)\right\Vert ,\label{eq:t i plus}\\
\vect q_{i,t_{i}^{+}}^{(1)} & = & \vect q_{j,t}^{(1)}-\vect I\left(\frac{1}{e_{j}e_{i}}\dot{\vect p}_{j,t}^{(1)}-\vect F\left(\vect q_{j,t}^{(1)}-\vect q_{i,t_{i}^{-}(t,\vect q_{j,t})}^{(0)}\right)\right)\label{eq:q i plus}
\end{eqnarray}
for all $t\in(t_{j}^{(1)},t_{j}^{(2)}]\cap[T_{j}^{\min},T_{j}^{\max}]$
and furthermore
\[
t_{i}^{(2)}=\min\left\{ t_{i}^{+}|_{t=t_{j}^{(2)}},T_{i}^{\max}\right\} .
\]
As in the first step one finds that 
\[
t\mapsto\begin{cases}
\vect q_{i,t}^{(0)} & \text{for }t\in[t_{i}^{(0)},t_{i}^{(1)}]\\
\vect q_{i,t}^{(1)} & \text{for }t\in(t_{i}^{(1)},t_{i}^{(2)}]
\end{cases}
\]
is smooth for $t\in[t_{i}^{(0)},t_{i}^{(2)}]$. Finally, due to (\ref{eq:no singularity})
we can define
\begin{equation}
\vect p_{i,t}^{(1)}:=m\frac{\dot{\vect q}_{i,t}}{\sqrt{1-\dot{\vect q}_{i,t}^{2}}},\qquad\forall\, t\in(t_{i}^{(1)},t_{i}^{(2)}].\label{eq:p i 1}
\end{equation}

In consequence, the maps
\[
t\mapsto(\vect q_{i,t},\vect p_{i,t}):=\begin{cases}
(\vect q_{i,t}^{(0)},\vect p_{i,t}^{(0)}) & \text{for }t\in[t_{i}^{(0)},t_{i}^{(1)}]\\
(\vect q_{i,t}^{(0)},\vect p_{i,t}^{(0)}) & \text{for }t\in(t_{i}^{(1)},t_{i}^{(2)}]
\end{cases}
\]
for $i=1,2$ are smooth, and by virtue of definitions (\ref{eq:p j 1}),(\ref{eq:p i 1})
and (\ref{eq:q j plus}), (\ref{eq:q i plus}) and (\ref{eq:fulfilled i}),(\ref{eq:fulfilled j})
of (\ref{lem:key_ingredient}) they fulfill
\begin{equation}
\begin{pmatrix}\vect q_{i,t}\\
\vect p_{i,t}
\end{pmatrix}=\begin{pmatrix}\vect v(\vect p_{i,t})\\
e_{i}e_{j}\left[\vect F\left(\vect q_{i,t}-\vect q_{j,t_{j}^{+}}\right)+\vect F\left(\vect q_{i,t}-\vect q_{j,t_{j}^{-}}\right)\right]
\end{pmatrix}\label{eq:toy model i}
\end{equation}
for $t\in[t_{i}^{(0)},t_{i}^{-}(t_{j}^{(2)},\vect q_{j,t_{j}^{(2)}})]$
and\foreignlanguage{english}{
\begin{equation}
\begin{pmatrix}\vect q_{j,t}\\
\vect p_{j,t}
\end{pmatrix}=\begin{pmatrix}\vect v(\vect p_{j,t})\\
e_{j}e_{i}\left[\vect F\left(\vect q_{j,t}-\vect q_{i,t_{i}^{+}}\right)+\vect F\left(\vect q_{j,t}-\vect q_{i,t_{i}^{-}}\right)\right]
\end{pmatrix}\label{eq:toy model j}
\end{equation}
}for $t\in[t_{j}^{(1)},t_{j}^{-}(t_{i}^{(2)},\vect q_{i,t_{i}^{(2)}})].$\\

This construction can be repeated where in the $k$th step one constructs
the extension
\[
(t_{i}^{(k)},t_{i}^{(k+1)}]:t\mapsto(\vect q_{i,t}^{(k)},\vect p_{i,t}^{(k)})_{i=1,2}.
\]
For each step one finds
\[
t_{i}^{(k+1)}\geq\min\left\{ t^{(k)}+d,T_{i}^{\max}\right\} ,\qquad i\in\{1,2\}.
\]
In consequence, only finite repetitions of this construction are needed
to compute
\begin{equation}
[t_{i}^{(0)},T_{i}^{\max}]:t\mapsto(\vect q_{i,t},\vect p_{i,t}):=\begin{cases}
(\vect q_{i,t}^{(0)},\vect p_{i,t}^{(0)}) & \text{for }t\in[t_{i}^{(0)},t_{i}^{(1)}]\\
(\vect q_{i,t}^{(k)},\vect p_{i,t}^{(k)}) & \text{for }t\in(t_{i}^{(k)},t_{i}^{(k+1)}]
\end{cases}\label{eq:steps}
\end{equation}
which fulfills (\ref{eq:toy model i}) for $t\in[t_{i}^{(0)},t_{i}^{-}(T_{j}^{\max},\vect q_{j,T_{j}^{\max}})]$
and (\ref{eq:toy model j}) for $t\in[t_{j}^{(1)},t_{j}^{-}(T_{i}^{\max},\vect q_{i,T_{i}^{\max}})]$.
The same construction can be carried out into the past which results
in smooth maps
\[
[T_{i}^{\min},T_{i}^{\max}]\to\mathbb{R}^{3}\times\mathbb{R}^{3},\qquad t\mapsto(\vect q_{i,t},\vect p_{i,t}),\qquad i\in\{1,2\}.
\]

From this construction we infer the claims of \thmref{uniquenss}:
Claim (i) follows from definition (\ref{eq:steps}). Furthermore,
due to \lemref{key_ingredient}, (\ref{eq:t j plus})-(\ref{eq:q j plus})
and (\ref{eq:t i plus})-(\ref{eq:q i plus}) the map $t\mapsto(\vect q_{i,t},\vect p_{i,t})$
fulfills (\ref{eq:toy model}) for times 
\[
t\in[t_{i}^{+}(T_{j}^{\min},\vect q_{j,T_{j}^{\min}}),t_{i}^{-}(T_{j}^{\max},\vect q_{j,T_{j}^{\max}})]
\]
for $i\in\{1,2\}$ and $j\neq i$ and therefore claim (ii). Finally,
\lemref{key_ingredient} guarantees that this constructed solution
is unique which proves claim (iii) and concludes the proof.
\end{proof}
As a byproduct we observe that specification of initial positions
and momenta of the charges as suggested by the work \cite{driver_canfuture_1979}
does not always ensure uniqueness:
\begin{cor}
\label{cor:newtonian}Let $\left(\vect q_{i}^{(0)},\vect p_{i}^{(0)}\right)_{i=1,2}$
be initial data such that there is a $t^{*}\in(t_{i}^{(0)},t_{i}^{(1)})\cap(t_{j}^{(0)},t_{j}^{(1)})$,
let $D_{i}\to\mathbb{R}^{3}\times\mathbb{R}^{3},t\mapsto(\vect q_{i,t},\vect p_{i,t})$
for $i\in\{1,2\}$ be the corresponding solution to (\ref{eq:toy model}),
and let $\vect q_{i}^{(0)},\vect p_{i}^{(0)}\in\mathbb{R}^{3}$ for
$i=1,2$ be defined as
\[
(\vect q_{i}^{(0)},\vect p_{i}^{(0)}):=(\vect q_{i,t},\vect p_{i,t})\big|_{t=t^{*}},\qquad i\in\{1,2\}.
\]
There are uncountably many other solutions $\widetilde{D}_{i}\to\mathbb{R}^{3}\times\mathbb{R}^{3},t\mapsto(\widetilde{\vect q}_{i,t},\widetilde{\vect p}_{i,t})$
for $i=1,2$ to (\ref{eq:toy model}) which fulfill
\begin{equation}
(\widetilde{\vect q}_{i,t},\widetilde{\vect p}_{i,t})\big|_{t=t^{*}}=(\vect q_{i}^{(0)},\vect p_{i}^{(0)}),\qquad i\in\{1,2\},\label{eq:newtonian}
\end{equation}
but not
\begin{equation}
(\widetilde{\vect q}_{i,t},\widetilde{\vect p}_{i,t})=(\vect q_{i,t},\vect p_{i,t})\qquad\forall\, t\in D_{i}\cap\widetilde{D}_{i},\quad i\in\{1,2\}.\label{eq:equality}
\end{equation}
\end{cor}
\begin{proof}
Choose $s\in\mathbb{R}$ and $\delta>0$ such that $(s-\delta,s+\delta)\subset(t_{i}^{(0)},t_{i}^{(1)})\cap(t_{j}^{(0)},t_{j}^{(1)})$
and \foreignlanguage{english}{$t^{*}\notin(s-\delta,s+\delta)$}.
Furthermore, for $\lambda>0$ and $i=1,2$ let $\vect d_{\lambda,i}:\mathbb{R}\to\mathbb{R}^{3}$,
$t\mapsto\vect d_{\lambda,i}(t)$ be a smooth function such that
\[
\mathrm{supp}\,\vect d_{\lambda,i}=[s-\delta,s+\delta],\qquad\sup_{t\in\mathbb{R}}|\dot{\vect d}_{\lambda,i}(t)|=\lambda.
\]
We define
\[
t\mapsto\widetilde{\vect q}_{i,t}^{(0)}:=\vect q_{i,t}^{(0)}+\vect d_{\lambda,i}(t),\qquad\forall\, t\in[t_{i}^{(0)},t_{i}^{(1)}],
\]
choose the parameter $\lambda>0$ such that $\left\Vert \dot{\widetilde{\vect q}}_{i,t}\right\Vert <1$,
and define
\[
t\mapsto\widetilde{\vect p}_{i,t}^{(0)}:=\frac{m_{i}\dot{\widetilde{\vect q}}_{i,t}}{\sqrt{1-\left\Vert \dot{\widetilde{\vect q}}_{i,t}\right\Vert ^{2}}},\qquad\forall\, t\in[t_{i}^{(0)},t_{i}^{(1)}].
\]
The maps $(\widetilde{\vect q}_{i}^{(0)},\widetilde{\vect p}_{i}^{(0)})_{i=1,2}$
are initial data according to \defref{initial_data} which fulfill
(\ref{eq:newtonian}). However, according to \thmref{uniquenss} the
solution \foreignlanguage{english}{$\widetilde{D}_{i}\to\left(\mathbb{R}^{3}\times\mathbb{R}^{3}\right),t\mapsto(\widetilde{\vect q}_{i,t},\widetilde{\vect p}_{i,t})$}
for $i=1,2$ to (\ref{eq:toy model}) corresponding to $(\widetilde{\vect q}_{i}^{(0)},\widetilde{\vect p}_{i}^{(0)})_{i=1,2}$
does not fulfill (\ref{eq:equality}). Note that there are uncountably
many choices, e.g. in $s,\delta,\lambda$ and $\vect d_{\lambda,i}$,
to define other $(\widetilde{\vect q}_{i}^{(0)},\widetilde{\vect p}_{i}^{(0)})_{i=1,2}$
such that (\ref{eq:newtonian}) holds. None of the corresponding solutions
however fulfill (\ref{eq:equality}).
\end{proof}

\section{Constants of Motion\label{sec:Constants-of-Motion}}

In the following we define an energy functional for the approximate
model from which the general structure of constants of motion will
become apparent. Throughout this section we consider a solution 
\[
\mathbb{R}\supseteq D_{i}\to\mathbb{R}^{3}\times\mathbb{R}^{3},\qquad t\mapsto(\vect q_{i,t},\vect p_{i,t}),\qquad i\in\{1,2\},
\]
to (\ref{eq:toy model}); see \thmref{uniquenss}.
\begin{defn}
\label{def:energy} We define a map $H:D_{1}\times D_{2}\to\mathbb{R}^{+}$
by 
\begin{eqnarray*}
H(t_{1},t_{2}): & = & \sum_{i=1}^{2}\sqrt{\vect p_{i,t_{i}}^{2}+m_{i}^{2}}+\frac{1}{2}\sum_{i=1}^{2}e_{i}\sum_{j\neq i}e_{j}\sum_{\pm}\frac{1}{\left\Vert \vect q_{i,t_{i}}-\vect q_{j,t_{j}^{\pm}(t_{i})}\right\Vert }\\
 &  & +\frac{1}{2}\sum_{i=1}^{2}e_{i}\sum_{j\neq i}e_{j}\sum_{\pm}\int_{t_{i}}^{t_{i}^{\pm}(t_{j})}ds\,\vect F(\vect q_{i,s}-\vect q_{j,t_{j}^{\pm}(s)})\cdot\dot{\vect q}_{i,s}
\end{eqnarray*}
where for the delay functions we have used the short-hand notation
\[
t_{i}^{\pm}(t)\equiv t_{i}^{\pm}(t,\vect q_{j,t}),\qquad i\in\{1,2\},\quad j\neq i.
\]
We refer to $H$ as the \emph{energy functional of the system}. See
\figref{energy} for an example of which data is needed to define
this functional.

\begin{figure}
\begin{centering}
\hfill{}\includegraphics{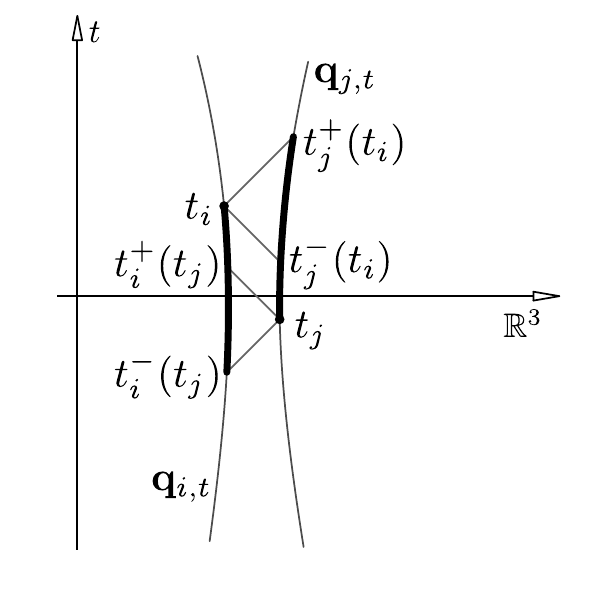}\hfill{}
\par\end{centering}

\caption{\label{fig:energy}Given $t_{i},t_{j}$ the thick lines denote the
data required to define the energy of the system.}

\end{figure}
\end{defn}
\begin{thm}
\label{thm:energy}For all $t_{1}\in D_{1}$ and \foreignlanguage{english}{\textup{$t_{2}\in D_{2}$}}
the equality
\[
H(t_{1},t_{2})=H(0,0)
\]
holds true.\end{thm}
\begin{proof}
Since we assume a smooth solution $t\mapsto(\vect q_{i,t},\vect p_{i,t})_{i=1,2}$
to (\ref{eq:toy model}) it is straight-forward to verify the claim
directly by differentiating with respect to $t_{1}$ and $t_{2}$.
Here, however, we want to provide a general idea how to find constants
of motion for a WF-type delay differential equations, and therefore
give a more instructive proof:\\

We start with the sum of the kinetic energy difference between times
$0$ and $t_{i}$ of the two particles $i=1,2$, that is
\begin{equation}
\sum_{i=1}^{2}\int_{0}^{t_{i}}ds\,\dot{\vect p}_{i,s}\cdot\vect v(\vect p_{i,s})\label{eq:kinetic}
\end{equation}
and make use of the equations (\ref{eq:toy model}) to express this
entity in terms of
\begin{equation}
(\ref{eq:kinetic})=\sum_{i=1}^{2}e_{i}\sum_{j\neq i}e_{j}\sum_{\pm}\int_{0}^{t_{i}}ds\,\vect F(\vect q_{i,s}-\vect q_{j,t_{j}^{\pm}(s)})\cdot\dot{\vect q}_{i,s}.\label{eq:kinetic2}
\end{equation}
It is convenient to split (\ref{eq:kinetic2}) into the following
summands:
\begin{eqnarray}
(\ref{eq:kinetic2}) & = & \frac{1}{2}\sum_{i=1}^{2}e_{i}\sum_{j\neq i}e_{j}\sum_{\pm}\int_{0}^{t_{i}}ds\,\vect F(\vect q_{i,s}-\vect q_{j,t_{j}^{\pm}(s)})\cdot\dot{\vect q}_{i,s}\label{eq:cancel_this}\\
 &  & +\frac{1}{2}\sum_{i=1}^{2}e_{i}\sum_{j\neq i}e_{j}\sum_{\pm}\int_{0}^{t_{i}}ds\,\vect F(\vect q_{i,s}-\vect q_{j,t_{j}^{\pm}(s)})\cdot\left(\dot{\vect q}_{i,s}-\dot{\vect q}_{j,t_{j}^{\pm}(s)}\frac{dt_{j}^{\pm}(s)}{ds}\right)\label{eq:potential}\\
 &  & +\frac{1}{2}\sum_{i=1}^{2}e_{i}\sum_{j\neq i}e_{j}\sum_{\pm}\int_{0}^{t_{i}}ds\,\vect F(\vect q_{i,s}-\vect q_{j,t_{j}^{\pm}(s)})\cdot\dot{\vect q}_{j,t_{j}^{\pm}(s)}\frac{dt_{j}^{\pm}(s)}{ds}.\label{eq:functional}
\end{eqnarray}
The integrand in (\ref{eq:potential}) is an exact differential so
that
\begin{equation}
(\ref{eq:potential})=-\frac{1}{2}\sum_{i=1}^{2}e_{i}\sum_{j\neq i}e_{j}\sum_{\pm}\frac{1}{\left\Vert \vect q_{i,t_{i}}-\vect q_{j,t_{j}^{\pm}(t_{i})}\right\Vert }+C\label{eq:pot_int}
\end{equation}
where $C\in\mathbb{R}$ is a constant. Next, we exploit the symmetries
of the force field and the delay function
\begin{equation}
\vect F(\vect x)=-\vect F(\vect x),\qquad\forall\,\vect x\in\mathbb{R}^{3}\setminus\{0\},\label{eq:sym_force}
\end{equation}
\begin{equation}
t=t_{i}^{\mp}\left(t_{j}^{\pm}(t)\right),\qquad i\in\{1,2\},\quad j\neq i.\label{eq:sym_delay}
\end{equation}
Now (\ref{eq:sym_delay}) allows to rewrite (\ref{eq:functional})
by substitution of the integration variable according to
\begin{equation}
(\ref{eq:functional})=\frac{1}{2}\sum_{i=1}^{2}e_{i}\sum_{j\neq i}e_{j}\sum_{\pm}\int_{t_{j}^{\pm}(0)}^{t_{j}^{\pm}(t_{i})}ds\,\vect F(\vect q_{i,t_{i}^{\mp}(s)}-\vect q_{j,s})\cdot\dot{\vect q}_{j,s}.\label{eq:subst}
\end{equation}
Furthermore, we apply (\ref{eq:sym_force}) and after that relabel
the indices $i\leftrightarrows j$ and $\pm\leftrightarrows\mp$ to
get
\begin{eqnarray}
(\ref{eq:subst}) & =- & \frac{1}{2}\sum_{i=1}^{2}e_{i}\sum_{j\neq i}e_{j}\sum_{\pm}\int_{t_{j}^{\pm}(0)}^{t_{j}^{\pm}(t_{i})}ds\,\vect F(\vect q_{j,s}-\vect q_{i,t_{i}^{\mp}(s)})\cdot\dot{\vect q}_{j,s}\nonumber \\
 & =- & \frac{1}{2}\sum_{i=1}^{2}e_{i}\sum_{j\neq i}e_{j}\sum_{\pm}\int_{t_{i}^{\mp}(0)}^{t_{i}^{\mp}(t_{j})}ds\,\vect F(\vect q_{i,s}-\vect q_{j,t_{j}^{\pm}(s)})\cdot\dot{\vect q}_{i,s}\nonumber \\
 & =- & \frac{1}{2}\sum_{i=1}^{2}e_{i}\sum_{j\neq i}e_{j}\sum_{\pm}\int_{t_{i}^{\pm}(0)}^{0}ds\,\vect F(\vect q_{i,s}-\vect q_{j,t_{j}^{\pm}(s)})\cdot\dot{\vect q}_{i,s}\label{eq:just_const}\\
 &  & -\frac{1}{2}\sum_{i=1}^{2}e_{i}\sum_{j\neq i}e_{j}\sum_{\pm}\int_{0}^{t_{i}}ds\,\vect F(\vect q_{i,s}-\vect q_{j,t_{j}^{\pm}(s)})\cdot\dot{\vect q}_{i,s}\label{eq:cancel}\\
 &  & -\frac{1}{2}\sum_{i=1}^{2}e_{i}\sum_{j\neq i}e_{j}\sum_{\pm}\int_{t_{i}}^{t_{i}^{\mp}(t_{j})}ds\,\vect F(\vect q_{i,s}-\vect q_{j,t_{j}^{\pm}(s)})\cdot\dot{\vect q}_{i,s}\label{eq:funct_term}
\end{eqnarray}
Noting that term (\ref{eq:just_const}) is just another constant and
term (\ref{eq:cancel}) cancels the first term on the right-hand side
of (\ref{eq:cancel_this}), the kinetic energy can be written as
\[
(\ref{eq:kinetic})=-\frac{1}{2}\sum_{i=1}^{2}e_{i}\sum_{j\neq i}e_{j}\sum_{\pm}\frac{1}{\left\Vert \vect q_{i,t_{i}}-\vect q_{j,t_{j}^{\pm}(t_{i})}\right\Vert }-\frac{1}{2}\sum_{i=1}^{2}e_{i}\sum_{j\neq i}e_{j}\sum_{\pm}\int_{t_{i}}^{t_{i}^{\mp}(t_{j})}ds\,\vect F(\vect q_{i,s}-\vect q_{j,t_{j}^{\pm}(s)})\cdot\dot{\vect q}_{i,s}+C
\]
which proves the claim.
\end{proof}

\section{Stable Orbits\label{sec:Stable-Orbits}}

By the symmetry of the advanced and retarded delay it is straight-forward
to construct stable orbits (compare \cite{schild_electromagnetic_1963}):
\begin{defn}
\label{def:schild}For masses $m_{1},m_{2}>0$, charges $e_{1},e_{2}\in\mathbb{R}$,
radii $r_{1},r_{2}>0$, and angular velocity $\omega\in\mathbb{R}\setminus\{0\}$
such that
\[
e_{1}e_{2}<0,\qquad\frac{m_{1}}{m_{2}}=\frac{\gamma(\omega r_{2})}{\gamma(\omega r_{1})}\frac{r_{2}}{r_{1}},\qquad0<\omega\Delta T<\frac{\pi}{2},\qquad\Delta T:=\sqrt{\frac{r_{1}^{2}+r_{2}^{2}}{1+\frac{r_{1}r_{2}}{e_{1}e_{2}}m_{1}\gamma(\omega r_{1})r_{1}\omega^{2}}},
\]
where
\[
\gamma(v):=\frac{1}{\sqrt{1-v^{2}}},
\]
we define the particle trajectories
\begin{equation}
t\mapsto\vect q_{1,t}:=r_{1}\begin{pmatrix}\cos\omega t\\
\sin\omega t
\end{pmatrix},\qquad t\mapsto\vect q_{2,t}:=-r_{2}\begin{pmatrix}\cos\omega t\\
\sin\omega t
\end{pmatrix}\label{eq:schild}
\end{equation}
which we call \emph{Schild solutions; }see \figref{orbits}.\end{defn}
\begin{thm}
\label{thm:schild}Given $m_{1,}m_{2},r_{1}>0$, $e_{1}\in\mathbb{R}$
there are $r_{2}>0$ and $\omega\in\mathbb{R}$ such that the Schild
solutions (\ref{eq:schild}) obey (\ref{eq:toy model}).

\begin{figure}
\begin{centering}
\hfill{}\includegraphics{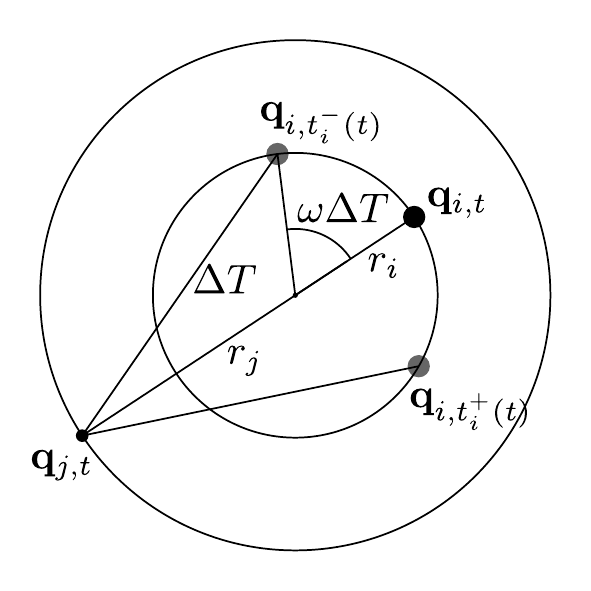}\hfill{}
\par\end{centering}

\caption{\label{fig:orbits}Stable orbit of two charges revolving around the
same center with radii $r_{i},r_{j}$. $\Delta T$ denotes the modulus
of the delay.}

\end{figure}
\end{thm}
\begin{proof}
We start by rewriting (\ref{eq:toy model}) as a second order equation,
i.e.
\[
\frac{d}{dt}\left(m_{i}\gamma(\dot{\vect q}_{i,t})\dot{\vect q}_{i,t}\right)=e_{i}e_{j}\left[\vect F\left(\vect q_{i,t}-\vect q_{j,t_{j}^{+}(t,\vect q_{i,t})}\right)+\vect F\left(\vect q_{i,t}-\vect q_{j,t_{j}^{-}(t,\vect q_{i,t})}\right)\right],\qquad i\in\{1,2\},\quad i\neq j.
\]
If the charges stay on the circular orbits (\ref{eq:schild}) then
velocities and accelerations are constant and fulfill
\[
\left\Vert \dot{\vect q}_{i,t}\right\Vert =\omega r_{i},\qquad\left\Vert \ddot{\vect q}_{i,t}\right\Vert =\omega^{2}r_{i}.
\]
Furthermore, the kinematics dictate that the net force acting upon
the particles must be centripetal w.r.t. the origin and equal, i.e.
\[
\left\Vert \frac{d}{dt}\left(m_{1}\gamma(\dot{\vect q}_{1,t})\dot{\vect q}_{1,t}\right)\right\Vert =m_{1}\gamma(\omega r_{1})\omega^{2}r_{1}=m_{2}\gamma(\omega r_{2})\omega^{2}r_{2}
\]
which implies
\begin{equation}
\frac{m_{1}}{m_{2}}=\frac{\gamma(\omega r_{2})}{\gamma(\omega r_{1})}\frac{r_{2}}{r_{1}}.\label{eq:mass_radius_relation}
\end{equation}

Next, we compute the delay function $\Delta T>0$ which, according
to (\ref{eq:wf_delay}), was defined by
\begin{equation}
\Delta T^{2}:=\left(t_{j}^{\pm}(t)-t\right)^{2}=\left\Vert \vect q_{i,t}-\vect q_{j,t_{j}^{\pm}(t)}\right\Vert ^{2}=r_{1}^{2}+r_{2}^{2}+2r_{1}r_{2}\cos\left(\omega\Delta T\right).\label{eq:delay}
\end{equation}
Due to the symmetry in the advanced and retarded delay (see \figref{orbits})
the net force is centripetal w.r.t. the origin for
\[
e_{1}e_{2}<0
\]
and its modulus equals
\[
\left\Vert e_{i}e_{j}\left[\vect F\left(\vect q_{i,t}-\vect q_{j,t_{j}^{+}(t,\vect q_{i,t})}\right)+\vect F\left(\vect q_{i,t}-\vect q_{j,t_{j}^{-}(t,\vect q_{i,t})}\right)\right]\right\Vert =\left|e_{1}e_{2}\frac{2\cos\left(\omega\Delta T\right)}{\Delta T^{2}}\right|.
\]
This means in particular that
\[
m_{1}\gamma(\omega r_{1})\omega^{2}r_{1}=\left|e_{1}e_{2}\frac{2\cos\left(\omega\Delta T\right)}{\Delta T^{2}}\right|
\]
which by (\ref{eq:wf_delay}) allows to solve for $\Delta T$ according
to
\begin{equation}
\Delta T^{2}=\frac{r_{1}^{2}+r_{2}^{2}}{1+\frac{r_{1}r_{2}}{e_{1}e_{2}}m_{1}\gamma(\omega r_{1})r_{1}\omega^{2}}\label{eq:delay_squared}
\end{equation}
for certain values of $r_{2}$ and $\omega$.

We can choose $e_{2}\in\mathbb{R}$ such that $e_{1}e_{2}<0$, $r_{2}>0$
such that (\ref{eq:mass_radius_relation}) holds true, and $\omega\in\mathbb{R}\setminus\{0\}$
with $\left|\omega\right|$ sufficiently small such that (\ref{eq:delay_squared})
is well-defined and $0<\omega\Delta T<\frac{\pi}{2}$. This proves
the claim.
\end{proof}

\section{What can we learn from the approximate model?\label{sec:Conclusion}}

We emphasize that the presented construction of solutions relies sensitively
on the simplicity of the force field $\vect F$ in (\ref{eq:force field}),
which was chosen such that it has a global inverse $\vect I$ as defined
in \lemref{key_ingredient}. Any generalization of this force field
that does not have a global inverse, in particular the one of WF,
will require a new technique. However, the approximate model uses
the same delay function (\ref{eq:wf_delay}) that is also used in
WF. Therefore one can expect that many mathematical structures appearing
in the approximate model will also arise in WF. As examples we point
out that the data needed to define the energy functional discussed
in \secref{Constants-of-Motion} is exactly the same as the one needed
to define the corresponding energy functional in WF (compare figure
3 in \cite{wheeler_classical_1949}), and the stable orbits actually
coincide with the ones in WF (compare \cite{schild_electromagnetic_1963}).
Given these similarities it seems reasonable to expect that also in
WF the natural choice of initial data that uniquely identifies solutions
will be of the same type as in the approximate model given in \defref{initial_data},
i.e. \figref{Initial-data.}. In this respect it is comforting to
note that the considered initial data is already sufficient to compute
the energy functional. The only additional information in the initial
data considered here, i.e. (\ref{eq:fulfilled i}) and (\ref{eq:fulfilled j}),
is about how smooth the solutions are. Smoothness will become a more
delicate issue when considering the full WF interaction as (\ref{eq:LW_E})
also depends on the acceleration.

Concerning the dynamics of many charges the question of typical behavior
becomes relevant and a generalization to many particles of the approximate
model is a good candidate for studying measures of typicality for
delay dynamics of this kind. Note that while the generalization of
the above uniqueness result requires a slightly more sophisticated
proof the results of \secref{Constants-of-Motion} and \secref{Stable-Orbits}
have straight-forward generalizations to many charges.

\bibliographystyle{alpha}

\begin{thebibliography}{SMSM71}

\bibitem[Bau97]{bauer_ein_1997}
G.~Bauer.
\newblock {\em Ein Existenzsatz f\"ur die {Wheeler-Feynman-Elektrodynamik}}.
\newblock Herbert Utz Verlag, 1997.

\bibitem[BD01]{bauer_maxwell-lorentz_2001}
G.~Bauer and D.~D\"urr.
\newblock The {Maxwell-Lorentz} system of a rigid charge.
\newblock {\em Annales Henri Poincare}, 2(1):179--196, April 2001.

\bibitem[BDD10]{bauer_wheeler_2010}
G.~Bauer, D.-A. Deckert, and D.~D\"urr.
\newblock On the {E}xistence of {D}ynamics of {W}heeler-{F}eynman
  {E}lectromagnetism.
\newblock {\em arXiv:1009.3103v2 [math-ph]}, 2010.

\bibitem[Dec10]{deckert_electrodynamic_2010}
D.-A. Deckert.
\newblock {\em Electrodynamic {A}bsorber {T}heory - {A} {M}athematical
  {S}tudy}.
\newblock Der Andere {Verlag}, January 2010.

\bibitem[Dir38]{dirac_classical_1938}
P.~A.~M. Dirac.
\newblock Classical theory of radiating electrons.
\newblock {\em Proceedings of the Royal Society of London. Series A,
  Mathematical and Physical Sciences}, 167(929):148--169, August 1938.

\bibitem[Dri79]{driver_canfuture_1979}
R.~D. Driver.
\newblock Can the future influence the present?
\newblock {\em Physical Review D}, 19(4):1098, February 1979.

\bibitem[Sch63]{schild_electromagnetic_1963}
A.~Schild.
\newblock Electromagnetic {Two-Body} problem.
\newblock {\em Physical Review}, 131(6):2762, 1963.

\bibitem[SMSM71]{siegel_lectures_1971}
{C.L.} Siegel, and {J.K.} Moser.
\newblock {\em Lectures on Celestial Mechanics}.
\newblock Springer, January 1971.

\bibitem[WF45]{wheeler_interaction_1945}
J.~A. Wheeler and R.~P. Feynman.
\newblock Interaction with the absorber as the mechanism of radiation.
\newblock {\em Reviews of Modern Physics}, 17(2-3):157, April 1945.

\bibitem[WF49]{wheeler_classical_1949}
J.~A. Wheeler and R.~P. Feynman.
\newblock Classical electrodynamics in terms of direct interparticle action.
\newblock {\em Reviews of Modern Physics}, 21(3):425, July 1949.

\end{thebibliography}

\vskip2cm
\noindent\emph{Dirk - Andr\'e Deckert}\\
Department of Mathematics, University of California Davis,\\
One Shields Avenue, Davis, California 95616, USA\\
\texttt{deckert@math.ucdavis.edu}

\vskip.5cm
\noindent\emph{Detlef D\"urr}\\
Mathematisches Institut der LMU München,\\
 Theresienstraße 39, 80333 München, Germany\\
\texttt{duerr@math.lmu.de}

\vskip.5cm
\noindent\emph{Nicola Vona}\\
Mathematisches Institut der LMU München,\\
 Theresienstraße 39, 80333 München, Germany\\
\texttt{vona@math.lmu.de}

\end{document}